\documentclass[11pt]{article}
\pdfoutput=1
\usepackage{microtype}

\usepackage[colorlinks=true,urlcolor=Blue,citecolor=Green,linkcolor=BrickRed]{hyperref}
\usepackage[usenames,dvipsnames]{xcolor}

\usepackage[margin=1in]{geometry}
\usepackage[utf8]{inputenc}
\usepackage[noadjust]{cite}
\usepackage[OT4]{fontenc}
\usepackage{amsthm}
\usepackage{amssymb}
\usepackage{amsmath}
\usepackage{amsfonts}
\usepackage{todonotes}
\usepackage{authblk}
\usepackage{rotating}
\usepackage{xspace}
\usepackage[inline]{enumitem}
\usepackage{verbatim}
\usepackage{mathrsfs}
\usepackage{mathtools}
\usepackage{framed}
\usepackage{algorithm}
\usepackage{algorithmicx}
\usepackage[noend]{algpseudocode}
\usepackage{algpseudocode}

\algnewcommand{\algorithmicand}{\textbf{ and }}
\algnewcommand{\algorithmicor}{\textbf{ or }}
\algnewcommand{\OR}{\algorithmicor}
\algnewcommand{\AND}{\algorithmicand}
\newcommand{\Oh}{\mathcal{O}}

\algnewcommand{\LineComment}[1]{\State \(\triangleright\) #1}

\newtheorem{theorem}{Theorem}[section]

\newtheorem{lemma}[theorem]{Lemma}
\newtheorem{proposition}[theorem]{Proposition}

\theoremstyle{definition}

\title{A Faster Subquadratic Algorithm for the Longest Common Increasing Subsequence Problem}
\author{Anadi Agrawal}
\author{Paweł Gawrychowski}
\date{}

\affil{Institute of Computer Science, University of Wrocław, Poland}

\newcommand{\problem}[3]{
\begin{framed}
  \noindent
  \textbf{Problem:} #1

  \noindent
  \textbf{Input:} #2

  \noindent
  \textbf{Output:} #3
\end{framed}
}

\newcommand{\countelem}{\textsf{cnt}}
\newcommand{\prev}{\textsf{prev}}

\begin{document}
\maketitle

\begin{abstract}
The Longest Common Increasing Subsequence (LCIS) is a variant of the classical Longest Common Subsequence (LCS),
in which we additionally require the common subsequence to be strictly increasing. While the well-known ``Four Russians
'' technique can be used to find LCS in subquadratic time, it does not seem applicable to LCIS.
Recently, Duraj [STACS 2020] used a completely different method based on the combinatorial properties
of LCIS to design an $\Oh(n^2(\log\log n)^2/\log^{1/6}n)$ time algorithm. We show that an approach
based on exploiting tabulation can be used to construct an asymptotically faster $\Oh(n^2 \log\log n/\sqrt{\log n})$ time algorithm.
As our solution avoids using the specific combinatorial properties of LCIS, it can be also adapted for the
Longest Common Weakly Increasing Subsequence (LCWIS).
 \end{abstract}


\section{Introduction}

In the well-known Longest Common Subsequence problem we aim to find the length of the longest subsequence common
to two strings $A[1..n]$ and $B[1..n]$. A textbook exercise is to find it in $\Oh(n^{2})$ time~\cite{WagnerF74}, and using
the so-called ``Four Russians'' technique this has been brought down to $\Oh(n^{2}/\log^{2}n)$ for constant alphabets~\cite{WagnerF74}
and $\Oh(n^{2}\log\log n/\log^{2}n)$ for general alphabets~\cite{BilleF08}. Recently, there was some progress in
providing explanation for why a strongly subquadratic $\Oh(n^{2-\epsilon})$ time algorithm is unlikely~\cite{AbboudBW15,BringmannK15},
and in fact even achieving $\Oh(n^{2}/\log^{7+\epsilon}n)$ would have some exciting unexpected consequences~\cite{AbboudB18}.
In this paper we consider a related problem defined as follows:
\problem{Longest Common Increasing Subsequence (LCIS)}
{integer sequences $A[1..n]$ and $B[1..n]$}
{largest $\ell$ such that there exist indices $i_{1}<\ldots <i_{\ell}$ and $j_{1}<\ldots<j_{\ell}$
with the property that (i) $A[i_{k}]=B[j_{k}]$, for every $k=1,\ldots,\ell$, and (ii) $A[i_{1}]< \ldots < A[i_{\ell}]$.}
While this is less obvious than for LCS, LCIS can be also solved in $\Oh(n^{2})$ time~\cite{YangHC05}
(and in linear space~\cite{Sakai06}),
and it can be proved that a strongly subquadratic algorithm would refute SETH~\cite{DurajKP19}
(although faster algorithms are known for some special cases~\cite{KutzBKK11}).
However, as opposed to LCS, the usual ``Four Russians'' approach, that roughly consists in partitioning the DP
table into blocks of size $\log n \times \log n$, doesn't seem directly applicable to LCIS.
Very recently, Duraj~\cite{Duraj20} used a completely different approach based on some nice combinatorial 
properties specific to LCIS to design a subquadratic $\Oh(n^{2}(\log\log n)^{2}/\log^{1/6}n)$ time algorithm.

\paragraph{Our contribution.}
We design a faster subquadratic $\Oh(n^{2}\log\log n/\sqrt{\log n})$ time algorithm for LCIS. Interestingly,
instead of using the combinatorial properties of LCIS as in the previous work we apply a technique based
on exploiting tabulation (but differently than in the classical ``Four Russians'' approach). This allows our algorithm to
be modified to solve the Longest Common Weakly Increasing Subsequence (LCWIS) problem
(for which an $\Oh(n^{2-\epsilon})$ time algorithm is also known to refute SETH~\cite{Polak18}). This doesn't
seem to be the case for Duraj's approach based on bounding the number of so-called significant symbol
matches, that for LCWIS might be $\Omega(n^{2})$.
Throughout the paper we assume that $A$ and $B$ are of the same length, and the goal is to calculate the
length of LCIS. However, the algorithm can be easily modified to avoid this assumption and recover the subsequence
itself.

\paragraph{Overview of the paper.}
Our algorithm is based on combining two different procedures. By appropriately selecting the parameters, the
overall complexity becomes $\Oh(n^{2}\log\log n/\sqrt{\log n})$ as explained in Section~\ref{sec:union}.

The first procedure described in Section~\ref{sec:first} works fast when there are only few distinct elements
in both sequences. We start with a solution based
on dynamic programming working in $\Oh(t\cdot n^{2})$ time, where $t$ is the number of distinct elements in both sequences.
Then, we exploit tabulation to decrease its running time to $\Oh(t\cdot n^{2}/\log n)$.

The second procedure described in Section~\ref{sec:second} is efficient when there are not too many matching pairs,
that is, pairs $(i,j)$ such that $A[i]=B[j]$.
The main idea is to calculate, for every such pair, LCIS of $A[1..i]$ and $B[1..j]$ that ends with $A[i]=B[j]$. This is done by applying an appropriate
dynamic predecessor structure. This roughly follows the ideas of Duraj, except that instead of using van Emde Boas trees
we notice that, in fact, one can plug in any balanced search trees with efficient split/merge.

In Section~\ref{sec:weakly} we explain the necessary modification required to adapt our solution for LCWIS.

\section{Preliminaries}

We work with sequences consisting of integers. For such a sequence $A$, we write $A[i]$ to denote
the $i$-th element, and $A[1..i]$ to denote the prefix of length $i$. $|A|$ is the length of $A$.
Let $\sigma$ be the sequence consisting of all distinct integers present in $A$ and $B$,
arranged in the increasing order, and $\countelem(v)$ be the total number of occurrences of $\sigma[v]$
in $A$ and $B$.

We call a pair of indices $(x,y)$ a \textit{matching pair} when $A[x]=B[y]$.
Further, we call it a $\textit{$\sigma[i]$-pair}$ when $A[x]=B[y]=\sigma[i]$.

We write $LCIS(i,j)$ to denote $LCIS(A[1..i],B[1..j])$, that is, the longest increasing common subsequence
of $A[1..i]$ and $B[1..j]$. 
We write $LCIS^{\rightarrow}(i, j)$ to denote the longest strictly increasing subsequence of $A[1..i]$ and $B[1..j]$
which includes both $A[i]$ and $B[j]$ (so in particular, $A[i]=B[j]$).

Throughout the paper, $\log x$ denotes $\log_{2} x$.

\section{First Solution}
\label{sec:first}

In this section we describe an algorithm for finding LCIS in $\mathcal{O}(|\sigma|\cdot n^2 / \log n)$ time.

Let $dp_{v}[i][j]$ denote the largest possible length of a sequence $C$ such that:
\begin{enumerate}
\item $C$ is an increasing common subsequence of $A[1..i]$ and $B[1..j]$,
\item $C$ consists of elements not larger than $\sigma[v]$.
\end{enumerate}
Then, our goal is to compute $dp_{|\sigma|}[n][n]$.

All $|\sigma|\cdot n^{2}$ entries in $dp$ can be calculated in $\Oh(1)$ time each using the following recurrence:
\[ 
	dp_{v + 1}[i][j] =
	\begin{dcases*} 
	\text{$\max\{ dp_v[i][j], dp_v[i - 1][j - 1] + 1\}, $} & if  $A[i] = B[j] = \sigma[v + 1]$, \\ 
	\text{$\max\{ dp_v[i][j], dp_{v + 1}[i - 1][j], dp_{v + 1}[i][j - 1]\}, $} & otherwise.
	\end{dcases*}
\]
In order to decrease the time we will speed up calculating $dp_{v+1}$ from $dp_{v}$.
Because calculating $dp_{v+1}$ only requires the knowledge of $dp_{v}$, we will only keep the current $dp_{v}$
and update all of its entries to obtain $dp_{v+1}$.

\begin{lemma}
\label{lem:unit}
$0 \le dp_{v}[i][j] - dp_{v}[i][j - 1] \le 1$ and $0 \le dp_{v}[i][j] - dp_{v}[i-1][j] \le 1$.
\end{lemma}

\begin{proof}
A subsequence of $B[1..(j - 1)]$ is still a subsequence of $B[1..j]$, so $dp_{v}[i][j-1] \le dp_{v}[i][j]$.
Consider a sequence $C$ corresponding to $dp_{v}[i][j]$, and let $C'$ be $C$ without the last element. 
Because $C$ is a subsequence of $B[1..j]$, $C'$ is a subsequence of $B[1..(j-1)]$. So, $C'$ is an increasing
subsequence of $A[1..i]$ and $B[1..(j-1)]$, hence $|C'| \le dp_{v}[i][j-1]$. As $|C|=|C'|+1$,
we conclude that $dp_{v}[i][j] \le dp_{v}[i][j-1] + 1$. The second part of the lemma follows by a symmetrical
reasoning.
\end{proof}

Instead of maintaining $dp_{v}$, we keep another table $dp'_{v}[i][j] = dp_{v}[i][j] - dp_{v}[i][j - 1]$ (where $dp_{v}[i][j] = 0$ for $j < 1$).
Due to Lemma~\ref{lem:unit}, each entry of $dp'_{v}$ is either 0 or 1. This allows us to store each row of $dp'_{v}$
by partitioning it into $\Oh(n / B)$ blocks of length $B$, with every block represented by a bitmask of size $B$
saved in a single machine word, where $B=\alpha\log n$ for some constant $\alpha$ to be fixed later.
By definition, $dp_{v}[i][j] = \sum\limits_{k = 1}^j dp'_{v}[i][k]$. In addition to $dp'_{v}$, we store the value of $dp_{v}[i][j]$
for every block boundary, so $\Oh(n^{2}/B)$ values overall. This will allow us later to recover any $dp_{v}[i][j]$ in constant
time by retrieving the value at the appropriate block boundary and adding the number of 1s in a prefix of some
bitmask. We preprocess such prefix sums for every possible bitmask in $\Oh(2^{B} \cdot B)$ time and space.

\begin{lemma}
\label{lem:upd}
$0 \le dp_{v + 1}[i][j] - dp_v[i][j] \le 1$.
\end{lemma}

\begin{proof}
Because allowing using more elements cannot decrease the length, $dp_v[i][j] \le dp_{v + 1}[i][j]$.
Let $C$ be a sequence corresponding to $dp_{v+1}[i][j]$, and let $C'$ be $C$ without the last element.
Because $C$ is strictly increasing and $\sigma$ consists of all distinct elements, the elements of
$C'$ are not larger than $\sigma(v)$, so $|C'| \le dp_{v}[i][j]$. Then, using $|C'|+1=|C|$ we obtain
that $dp_{v + 1}[i][j] - 1 \le dp_v[i][j]$.
\end{proof}

We now describe how to calculate $dp'_{v + 1}$. We start with describing an approach
that works in $\Oh(n^{2})$ time and then explain how to accelerate it to $\Oh(n^{2}/\log n)$.
We use the recursion for $dp_{v+1}[i][j]$ to update the rows of $dp'_{v + 1}$ one-by-one. While updating
the entries in a row going from left to right we are no longer guaranteed that $dp_{v + 1}[i][j] \le dp_{v + 1}[i][j + 1]$,
so $dp'_{v + 1}[i][j]$ can become negative. To overcome this issue, we immediately propagate each value to
the right: after increasing $dp_{v + 1}[i][j]$ (by one due to Lemma~\ref{lem:upd}) we also increase every
$dp_{v + 1}[i][k]$ equal to the original value of $dp_{v + 1}[i][j]$, for all $k>j$. This translates into setting
$dp'_{v + 1}[i][j]$ to 1 and setting $dp'_{v + 1}[i][k]$ to 0, for the smallest $k>j$ such that $dp'_{v + 1}[i][k]=1$. 
To implement this efficiently, we maintain $k$ while considering $j=1,2,\ldots,n$ in $\Oh(n)$ overall
time. The details of this procedure are shown in Algorithm~\ref{alg:row}.

\begin{algorithm}[H]
\begin{algorithmic}[1]
\Procedure{CalculateRow}{$v, i$}
	\State $ptr \gets 1$
	\State $cur\_value \gets 0$
	\State $prv\_value \gets 0$
	\State $prv\_phase \gets 0$
	
	\For{$j = 1..n$}
		\State $dp'_{v + 1}[i][j] = dp'_v[i][j]$
	\EndFor
	\For{$j = 1..n$}
		\If{$ptr \le i$} $ptr \gets i + 1$
		\EndIf

		\While{$ptr \le n \algorithmicand dp'_{v + 1}[i][ptr] = 0$}
			\State $ptr \gets ptr + 1$
		\EndWhile
		
		\State $cur\_value \gets cur\_value + dp'_{v + 1}[i][j]$
		\LineComment{$cur\_value = \sum_{j'=1}^{j}dp'_{v+1}[i][j'] = \max \{ dp_{v}[i][j],dp_{v+1}[i][j-1]\}$}
		\LineComment{$prv\_phase = dp_{v}[i-1][j-1]$}
		\If{$A[i] = B[j] =\sigma[v + 1] \algorithmicand cur\_value = prv\_phase$}
			\State $dp'_{v + 1}[i][j] \gets 1$
			\State $cur\_value \gets cur\_value + 1$
			\If{$ptr \le n$} 
				$dp'_{v + 1}[i][ptr] \gets 0$
			\EndIf
		\EndIf
		
		\State $prv\_phase \gets prv\_phase + dp'_v[i - 1][j]$
		\State $prv\_value \gets prv\_value + dp'_{v + 1}[i - 1][j]$
		\LineComment{$prv\_value = dp_{v+1}[i-1][j]$}
		
		\If{$cur\_value < prv\_value$}
			\State	$cur\_value \gets prv\_value$
			\State $dp'_{v + 1}[i][j] \gets 1$
			\If{$ptr \le n$} 
				$dp'_{v + 1}[i][ptr] \gets 0$
			\EndIf			
		\EndIf
	\EndFor
\EndProcedure
\end{algorithmic}
\caption{Calculate the $i$-th row of $dp'_{v + 1}$}
\label{alg:row}
\end{algorithm}

We speed up Algorithm~\ref{alg:row} by a factor of $B$ by considering whole blocks of $dp'_{v + 1}$ instead of single entries.
Consider a single block of $dp'_{v + 1}$ consisting of the values of $dp'_{v + 1}[i][j], dp'_{v + 1}[i][j + 1], \ldots, dp'_{v + 1}[i][j + B - 1]$, and assume
that they have been already partially updated by propagating the maximum. To calculate their correct values
we need the following information:
\begin{enumerate}
\item $dp'_v[i - 1][j], dp'_v[i-1][j+1], \ldots, dp'_v[i - 1][j + B - 1]$,
\item $dp'_{v + 1}[i - 1][j], dp'_{v + 1}[i-1][j+1], \ldots, dp'_{v + 1}[i - 1][j + B - 1]$,
\item $dp'_{v + 1}[i][j], dp'_{v + 1}[i][j+1], \ldots, dp'_{v + 1}[i][j + B - 1]$,
\item $dp_v[i - 1][j-1]$,
\item $dp_{v + 1}[i - 1][j - 1]$,
\item $dp_{v+1}[i][j-1]$,
\item for which indices $j,j+1,\ldots,j+B-1$ we have $A[i]=B[j]=\sigma[v + 1]$.
\end{enumerate}
In fact, we can rewrite the procedure so that instead of the values 
$dp_v[i - 1][j-1]$, $dp_{v+1}[i - 1][j - 1]$, $dp_{v+1}[i][j-1]$ 
only the differences $dp_{v+1}[i - 1][j-1]-dp_v[i - 1][j-1]$ and
$dp_{v+1}[i][j-1]-dp_{v+1}[i - 1][j-1]$ are needed. 
By Lemma~\ref{lem:unit} and Lemma~\ref{lem:upd}, both
differences belong to $\{0,1\}$, so the whole information required for calculating the correct values consists of
$4B+2$ bits. Blocks $dp'$ are already stored in separate machine words, and we can prepare, for every
$v$, an array with the $j$-th entry set to 1 when $B[j]=v$, partitioned into $n/B$ blocks of length $B$,
where each block is saved in a single machine word, in $\Oh(|\sigma|\cdot n)$ time. This allows us
to gather all the required information in constant time and use a precomputed table of size $\Oh(2^{4B + 2})$ 
that stores a single machine word encoding the correct values in a block for every possible combination.
Additionally, the table stores the number of 1s to the right of the block that should be changed to 0.
The table can be prepared in $\Oh(2^{4B+2}\cdot B)$ time by a straightforward modification of Algorithm~\ref{alg:row}.
Now we can update a whole block in constant time by retrieving the precomputed answer, but
then we still might need to remove some 1s on its right. Instead of removing them one-by-one we
work block-by-block. In more detail, we maintain a pointer to the nearest block that might contain
a 1. Let the number of 1s there be $\ell$ and the number of 1s that still need to be removed
be $s$. As long as $s>0$, we remove $\min\{\ell,s\}$ leftmost 1s from the current block in constant time
using a precomputed table of size $\Oh(2^{B}\cdot B)$, decrease $s$ by $\min\{\ell,s\}$, and move to
the next block. This amortises to constant time per block
over the row.

We set $B = \frac{\log n}{5}$ as to make the required preprocessing $o(n)$. Then, the overall complexity
of the algorithm becomes $\mathcal{O}(|\sigma|\cdot n^2 / \log n)$.

\section{Second Solution}
\label{sec:second}

In this section we describe an algorithm for solving LCIS in $\Oh(\sum\limits_{v = 1}^{|\sigma|}{(\countelem(v))^2(1+\log^2(n /\countelem(v))}))$
time.

For every matching pair $(x,y)$, we will compute $LCIS^{\rightarrow}(x, y)$, called the result for $(x,y)$.
The algorithm proceeds in phases corresponding to the elements of $\sigma$, and
in the $v$-th step computes the results for all $\sigma[v]$-pairs.
During this computation we maintain, for every $r=1,2,\ldots,n$, a structure $D(r)$ that allows us to quickly determine, given
any $(x,y)$, if there exists an already processed matching pair $(x',y')$ with result $r$ such that $x'<x$ and $y'<y$. Each $D(r)$ is
implemented using the following lemma.

\begin{lemma}
\label{lem:batched}
We can maintain a set of points $S\subseteq [n]\times [n]$ under inserting a batch of $u\leq n$ points in $\Oh(u(1+\log\frac{n}{u}))$
time and answering a batch of $q\leq n$ queries of the form ``given $(x,y)$, is there $(x',y')\in S$ such that $x'<x$ and $y'<y$'' in
$\Oh(q(1+\log\frac{n}{q}))$ time.
\end{lemma}

\begin{proof}
We observe that if the current $S$ contains two distinct points $(x_{i},y_{i})$ and $(x_{j},y_{j})$ with $x_{i}\leq x_{j}$ and $y_{i}\leq y_{j}$
then there is no need to keep $(x_{j},y_{j})$.  Thus, we keep in $S$ only points that are not dominated. Let $(x_{1},y_{1}),\ldots,(x_{k},y_{k})$
be these points arranged in the increasing order of $x$ coordinates (observe that we cannot have two non-dominated points with the same
$x$ coordinate). So, $x_{1}<x_{2}<\ldots < x_{k}$, where $k\leq n$, and because the points are not dominated also $y_{1}>y_{2}>\ldots y_{k}$.
We store the $x$ coordinates in a BST. This clearly allows us to answer a single query $(x,y)$ in $\Oh(\log n)$ time by locating
the predecessor of $x$. To insert a point $(x,y)$, we first check that it is not dominated by locating the predecessor
of $x$. Then, we might need to remove some of the subsequent $x$ coordinates that correspond to points that are dominated
by $(x,y)$. This can be efficiently implemented by maintaining a doubly-linked list of all points, and linking each $x$ coordinate
with its corresponding point. Insertion takes $\Oh(\log n)$ time plus another $\Oh(\log n)$ for every removed point, so $\Oh(\log n)$
amortised time, and a query concerning $(x,y)$ reduces to finding the predecessor of $x$ among the $x_{i}$s, which is
still too slow.

We use a BST that allows split and merge in $\Oh(\log s)$ time, where $s$ is the number of stored elements, for example AVL trees.
Additionally, we store the size of the subtree in every node. Then we have the following easy proposition.

\begin{proposition}
\label{pro:spl}
We can split BST into at most $b$ smaller BSTs containing $\Theta(s/b)$ elements each in $\Oh(b (1+\log {\frac{s}{b}}))$ time.
\end{proposition}

\begin{proof}
As long as there is a BST of size at least $2s/b$ we split it into two BSTs of (roughly) equal sizes.
Assuming for simplicity that both $s$ and $b$ are powers of 2, this takes $\Oh(\sum\limits_{i = 0}^{\log b - 1}2^i\log (s / 2^i))$
overall time, which can be bounded by calculating $\int_{1}^{b}\log(s/x)dx=\Oh(b(1+\log(s/b)))$.
\end{proof}

To process a batch of $b$ insertions/queries efficiently, we first sort them in $\Oh(b(1+\log (n / b)))$ time.
Then, we split the BST into at most $b$ smaller BSTs containing $\Theta(s/b)$ elements each, where $s$ is the number
of stored elements, using Proposition~\ref{pro:spl}. Because insertions/queries are sorted, we can determine for each
of them the relevant BST by a linear scan, and then insert/query the relevant BST in $\Oh(1+\log(s/b))$ time per operation
(if there are more than $s/b$ insertions to the same smaller BST, we split it into trees containing single elements, and
partition the insertions into groups of $\Theta(s/b)$).
Finally, we merge the BSTs into pairs, quadruples, and so on. By the calculation from the proof
of Proposition~\ref{pro:spl} this also takes $\Oh(b(1+\log(s/b)))$ time.
\end{proof}

Lemma~\ref{lem:batched} is already enough to binary search for the result of $(x,y)$ in $\Oh(\log^{2}n)$ time due to the following
property.

\begin{lemma}
\label{lem:bsearch}
Consider any $r$ and an already processed matching pair $(x',y')$ with result $r$. Then either $r=1$ or
there exists an already processed matching pair $(x'',y'')$ with result $r-1$ such that $x'' < x'$ and $y'' < y'$.
\end{lemma}

\begin{proof}
Assume that $r\geq 2$ and consider a sequence $C$ which realises the result for $(x',y')$.
Then $C[1..|C|-1]$ is an increasing subsequence of both $A[1..(x'-1)]$ and $B[1..(y'-1)]$.
Let $A[x'']$ and $B[y'']$ be its last elements in $A$ and $B$, respectively. Then $x'' < x'$, $y'' < y'$, and
$A[x'']=B[y'']$, so $(x'',y'')$ is a matching pair, and because $C$ is strictly increasing this matching pair
must have been already processed.
\end{proof}

However, our goal is to spend $\Oh(1+\log^{2}(n / \countelem(v)))$ time per every $(x,y)$. We exploit
the following property.

\begin{lemma}
\label{lem:monotone}
Consider two $\sigma[i]$ pairs $(x,y_{1})$ and $(x,y_{2})$, where $y_{1} < y_{2}$. The result for $(x,y_{2})$
is at least as large as for $(x,y_{1})$.
\end{lemma}

\begin{proof}
Consider a sequence $C$ which realises $LCIS^\rightarrow(x, y_1)$.
Then, replacing $y_{1}$ with $y_{2}$ we obtain a valid candidate for the value of $LCIS^\rightarrow(x, y_2)$.
\end{proof}

Consider all $\sigma[v]$ pairs with the same $x$ coordinate $(x,y_{1}),(x,y_{2}),\ldots,(x,y_{\countelem({\sigma[v]})})$.
We binary search for the result of $(x,y_{i})$ for $i=\countelem(v),\ldots,2,1$. By Lemma~\ref{lem:monotone},
in the $i$-th step we can start with the result found in the $(i+1)$-th step. Using doubling binary search,
by convexity of the log function the overall complexity becomes $\Oh(\countelem(v) (1+\log(n / \countelem(v))))$. This is still
too slow, as every step involves a separate invocation Lemma~\ref{lem:batched} and takes $\Oh(\log n)$ time.
To obtain the final speed up, we process all $x$ coordinates $x_{1},x_{2},\ldots,x_{\countelem(v)}$ together.
The high level idea is to synchronise all binary searches and exploit the possibility of asking a batch of queries.

We start with modifying the proof of Lemma~\ref{lem:batched} to allow for more general queries:
given $x$, we want to find the smallest $y$ such that there exists $(x',y')\in S$ with $x' < x$ and $y' < y$
(or detect that there is none). The modification is straightforward and doesn't increase the time complexity.
Now we can restate processing all pairs with the same $x$ coordinates.
We start with a counter $c$ initially set to $n$
and $i$ set to $\countelem(v)$. As long as $i\geq 1$, we use doubling binary search starting at $c$ to find the result
for $(x,y_{i})$. Let $c'$ be the found result. We use the modified Lemma~\ref{lem:batched}
to determine the smallest $y$ such that $c'$ is the result for $(x,y)$ and then keep decreasing $i$ as
long as $i\geq 1$ and $y_{i} > y$. Then, we decrease $c'$ by $1$ and repeat.

We further reformulate processing all pairs with the same $x$ coordinate. Consider a conceptual
complete binary tree on $n$ leaves (without losing generality, $n$ is a power of 2). Every node corresponds
to an interval $[a,b]$, and by querying such a node we will understand querying structure $D(a)$ with
the current $(x,y_{i})$. Consider the leaf corresponding to $c$. Calculating $c'$ with
doubling binary search can be phrased as starting at the leaf corresponding to $c$ and going up as
long as the query at the current node fails 
(we only need to ask a query if the previous node was
the right child of the current node; otherwise, we can immediately jump
to the nearest ancestor with such property).
After having reached the first ancestor
for which the query succeeds, we descend from its left child to the leaf corresponding to $c'$ by
repeating the following step: if querying the right child of the current node succeeds we descend
to the right child, and otherwise we descend to the left child.

Now we are able to synchronize the binary searches as follows. We traverse the conceptual complete
binary tree recursively: to traverse the subtree rooted at node $u$ with children $u_{\ell}$ and $u_{r}$
we (i) visit $u$, (ii) recursively traverse the subtree rooted at $u_{r}$, (iii) visit $u$ again, (iv) recursively traverse the
subtree rooted at $u_{\ell}$. Thus, every node is visited twice. We claim that
when visiting the nodes of the conceptual complete binary tree using this strategy, for any $x$
coordinate we are always able to wait till we encounter the node that should be queried next.
This is formalised in the following lemma.

\begin{lemma}
Let the result for $(x,y_{i+1})$ be $c$ and the result for $(x,y_{i})$ be $c'<c$. 
All queries necessary to calculate $c'$ can be answered during the traversal after the
second visit to $c$ and before the second visit to $c'$. 
\end{lemma}

\begin{proof}
The calculation consists of two phases. First, we need to ascend from the leaf corresponding
to $c$, reaching its first ancestor $u$ at which the query fails. Recall we only need to ask queries
if the previous node is the left child of the current node. For each such node $v$ we will be able
to use second visit to $v$ in the traversal. Thus, we will process all such queries after the second
visit to $u$. Then, we need to descend from the left child of $u$. In every step, we query
the right child $v_{r}$ of the current node $v$, and continue either in the left or in the right subtree of $v$.
To this end, we use the first visit to $v_{r}$ in the traversal.
\end{proof}

For each $x$ coordinate, by convexity of the log function, we need to query at most $\Oh(\countelem(v) (1+\log(n / \countelem(v))))$
nodes of the conceptual binary tree. Denoting by $q_{u}$ the number of queries to a node $u$,
we thus have $\sum_{u} q_{u} = s =
\Oh(\countelem(v)^{2}(1+\log (n / \countelem(v))))$. Invoking
Lemma~\ref{lem:batched}, the total time to answer all these queries is $\sum_{u} q_{u}(1+\log (n / q_{u}))$.
By convexity of the function $f(x)=x\log(n/x)$, this is maximised when all $q_{u}$s are equal,
but there are only $n$ of them, making the total time :
\[ \sum_{u} q_{u}(1+\log (n / q_{u})) \leq s(1+\log(n^{2}/s))  \leq s(1+\log(n^{2}/\countelem(v)^{2})) = \Oh(\countelem(v)^{2}(1+\log(n/\countelem(v)))^{2}) .\]

\section{Combining Solutions}
\label{sec:union}

Let $c$ be a parameter to be fixed later.
We call $\sigma[v]$ \textit{frequent} if $\frac{n}{c} < \countelem(v)$, and \textit{rare} otherwise.

We partition the sequence $\sigma$ into fragments. Each fragment is either a single frequent element or a maximal
range of rare elements. By definition of a frequent element and maximality of fragments consisting of rare elements,
we have $\Oh(c)$ fragments. We maintain the $dp_{v}$ table as in the first solution, but we only update it after having
processed a whole fragment. So, when considering a fragment starting at $\sigma[v]$ we only assume that the values
of $dp_{v-1}$ can be access in constant time. For a fragment consisting of a single frequent element, we proceed
exactly as in the first solution. In the remaining part of the description we describe how to process a fragment
consisting of rare elements $\sigma[v],\sigma[v+1],\ldots$.

We consider all $\sigma[v']$-pairs, for $v'=v,v+1,\ldots$. We will compute $LCIS^{\rightarrow}(x,y)$ for each such matching
pair $(x,y)$, and store it in the appropriate structure $D(r)$ implemented as described in Lemma~\ref{lem:batched}. 
To compute the values of $LCIS^{\rightarrow}(x,y)$ for all $\sigma[v']$-pairs, we use parallel binary search as
in the second solution with the following modification. To check if $LCIS^{\rightarrow}(x,y_{i}) > r$, we need to
consider two possibilities for the corresponding sequence $C$ ending at $A[x]=B[y_{i}]=\sigma[v']$:
\begin{enumerate}
\item If $C[|C|-1]$ belongs to the same fragment then it is enough to check if $D(r)$ contains a pair $(x',y')$ with
$x' < x$ and $y' < y_{i}$.
\item Otherwise, it is enough to check if $dp_{v-1}[x][y_{i}] \geq r$.
\end{enumerate}
Additionally, after having found $c'$ we need to keep decreasing $i$ as long as $i\geq 1$ and the answer for $(x,y_{i})$
is $c'$, and this needs to be tested in constant time per each such $i$. We again need to consider two possibilities,
and either compare $y_{i}$ with the value of $y'$ found by querying $D(c'-1)$ with $x$, or test if $dp_{v-1}[x][y_{i}] \geq r$
in constant time. Overall, this incurs only additional constant time per every step of the binary search for
every considered matching pair.

After having considered all $\sigma[v']$-pairs for the last element $\sigma[v']$ in the current fragment, we need to
compute $dp_{v'}$ from $dp_{v-1}$ and the calculated values of $LCIS^{\rightarrow}$. Of course, we want
to operate on $dp'_{v'}$ and $dp'_{v-1}$ instead of $dp_{v'}$ and $dp_{v-1}$. This is done row-by-row.
The $i$-th row is computed in two steps.

First, we need to set $dp_{v'}[i][j] = \max\{ dp_{v'}[i-1][j], dp_{v-1}[i][j]\}$
for every $j=1,2,\ldots,n$. This is done by processing whole blocks in constant time and precomputing the result for every possible
combination of the following information:
\begin{enumerate}
\item $dp'_{v'}[i-1][j],dp'_{v'}[i-1][j+1],\ldots,dp'_{v'}[i-1][j+B-1]$,
\item $dp'_{v-1}[i][j],dp'_{v-1}[i][j+1],\ldots,dp'_{v-1}[i][j+B-1]$,
\item $dp_{v'}[i-1][j-1]$,
\item $dp_{v-1}[i][j-1]$.
\end{enumerate}
This can be preprocessed in $\Oh(4^{B}\cdot B^{2})$ time after observing that, as in the first solution, only the difference
$dp_{v'}[i-1][j-1]-dp_{v-1}[i][j-1]$ is relevant and, additionally, it can be capped at $B$ (if it is bigger than $B$ then we can set it to $B$). The time is $\Oh(n/B)$.

Second, we need to consider the values of $LCIS^{\rightarrow}(i,j)$ computed for the current fragment.
If the result computed for a matching pair $(i,j)$ is $r$ then we need to update $dp_{v}[i][j'] = \max\{dp_{v}[i][j'], r\}$,
for every $j' \geq j$. This can be done by simultaneously scanning all such $j$s and the blocks. By maintaining
the maximum $r$, we can update the value of $dp_{v}[i][j]$ at the beginning of the block. Then, we
consider all other $j'$s belonging to the same block, and consider its corresponding result $r'$.
If $dp_{v}[i][j'] \geq r'$ then this result is irrelevant, and otherwise we must increase some of the values
in the block by 1 (as $dp_{v}[i][j'-1]$ is assumed to have been already updated and due to Lemma~\ref{lem:unit}).
As in the first solution, this is implemented by setting $dp'_{v}[i][j']=1$ and changing the nearest 1 into 0.
Overall, the time is bounded by the number of considered matching pairs plus additional $\Oh(n/B)$ time.

We set $B = \frac{\log n}{5}$ so that the preprocessing time is $o(n)$. For each frequent element we
spend $\Oh(n^{2}/B)$ time, so $\Oh(n^{2}/B\cdot c)$ overall. For each fragment consisting of rare elements,
the time is $\Oh(\countelem(v)^2 \log^2 (n / \countelem(v)))$ for every $v$ to compute the results,
and then $\Oh(n^{2}/B)$ plus the number of results. Using $\countelem(v)\leq n/c$, where $c$ is sufficiently large,
and calculating the derivative of $f(x)=x\log^{2}(n/x)$ we upper bound $\countelem(v)\log^{2}(n/\countelem(v)) \leq n/c\cdot\log^{2}c$
for every rare $v$, so the overall time is $\Oh(n^{2}/B\cdot c+n/c\cdot\log^{2}c\sum_{v}\countelem(v))=\Oh(n^{2}/B\cdot c+n^2/c\cdot\log^{2}c)$.

Choosing $c = \sqrt{\log n} \log \log n$ we obtain an algorithm working in $\Oh(n^2 \log \log n/\sqrt{\log n})$
time.

\section{Longest Common Weakly Increasing Subsequence}
\label{sec:weakly}

In this section we explain how to modify the algorithm to solve the weakly increasing version of the problem.
We adapt both solutions without changing their complexity as explained below, and then combine them using
the same threshold for the frequent/rare elements to arrive at $\Oh(n^2 \log \log n/\sqrt{\log n})$ complexity.

\subsection{First solution}

We define $dp$ as in the algorithm for LCIS. It can be calculated using the following recurrence (slightly different than
for LCIS):
\[ 
	dp_{v + 1}[i][j] =
	\begin{dcases*} 
	\text{$\max\{ dp_v[i][j], dp_{v + 1}[i - 1][j - 1] + 1\}, $} & if  $A[i] = B[j] = \sigma[v + 1]$, \\ 
	\text{$\max\{ dp_v[i][j], dp_{v + 1}[i - 1][j], dp_{v + 1}[i][j - 1]\}, $} & otherwise.
	\end{dcases*}
\]
The proof of Lemma~\ref{lem:unit} still holds, so we can store a table $dp'$ and retrieve any value of $dp$ from $dp'$
in constant time.

Algorithm~\ref{alg:row} stays essentially the same so we skip a detailed explanation. The speed up is implemented
by considering whole blocks of $dp'_{v + 1}$ instead of single entries. Consider a single block of $dp'_{v + 1}$ consisting of
the values of $dp'_{v + 1}[i][j], dp'_{v + 1}[i][j + 1], \ldots, dp'_{v + 1}[i][j + B - 1]$, and assume
that they have been already partially updated by propagating the maximum. To calculate their correct values
we need the following information:
\begin{enumerate}
\item $dp'_{v + 1}[i - 1][j], dp'_{v + 1}[i - 1][j+1], \ldots, dp'_{v + 1}[i - 1][j + B - 1]$,
\item $dp'_{v + 1}[i][j], dp'_{v + 1}[i][j+1], \ldots, dp'_{v + 1}[i][j + B - 1]$,
\item $dp_{v + 1}[i - 1][j-1]$,
\item $dp_{v+1}[i][j-1]$,
\item for which indices $j,j+1,\ldots,j+B-1$ we have $A[i]=B[j]=\sigma[v + 1]$.
\end{enumerate}
Once again we can rewrite the procedure so that instead of the values $dp_{v + 1}[i - 1][j - 1]$ and $dp_{v + 1}[i][j - 1]$ only
the difference $dp_{v + 1}[i][j - 1] - dp_{v + 1}[i - 1][j - 1]$ is needed. By Lemma~\ref{lem:unit}, the difference
belongs to $\{0, 1\}$, so the whole information required for calculating the correct values consists of $3B + 1$ bits.
This allows us to update the whole table in $\Oh(n^2 / B)$ as for LCIS.

We set $B = \frac{\log n}{4}$ as to make required preprocessing $o(n)$. Overall complexity of the algorithm
becomes $\Oh(|\sigma|n^2 / \log n)$.

\subsection{Second solution}

Calculating the result for each $\sigma[v]$-pair consists of two phases. 
In the first phase, for each $\sigma[v]$-pair $(x, y)$, we calculate the result 
assuming that all previous elements in the subsequence are strictly smaller than $\sigma[v]$. 
In the second phase, we calculate the result assuming that the previous element
is also equal to $\sigma[v]$.  The first phase can be implemented exactly as for LCIS
in $\Oh(\countelem(v)^{2}(1+\log^{2} (n / \countelem(v))))$ time. We now focus
on explaining how to implement the second phase.
Let $\prev_A[x]$ denote the greatest $x'$ fulfilling $A[x'] = A[x]$, if there is no such then $\prev_A[x] = 0$.
Similarly we define $\prev_B[y]$, both array can be prepared in negligible $\Oh(n\log n)$ time.

We analyze all $\sigma[v]$-pairs in the increasing order of rows and columns. 
Consequently, when analysing a pair~$(x, y)$, for all other $\sigma[v]$-pairs with 
$x' \leq x$, $y' \leq y$ we have already correctly calculated $LCWIS^{\rightarrow}(x', y')$. 
The proof of Lemma~\ref{lem:monotone} still holds for LCWIS, and implies that among
all other $\sigma[v]$-pairs $(x',y')$ such that $x'\leq x$ and $y'\leq y$ the
pair $(\prev_A[x], \prev_B[y])$ has the largest result. We can calculate $LCWIS^{\rightarrow}(x, y)$ as the maximum of the result computed in first phase and $LCWIS^{\rightarrow}(\prev_A[x], \prev_B[y]) + 1$.

The second phase takes only $\Oh(\countelem(v)^2)$ time, so the overall complexity remains
$\Oh(\countelem(v)^{2}(1+\log^{2} (n / \countelem(v))))$.

\section{Conclusions}

The $\Oh(n^2 \log \log n/\sqrt{\log n})$ complexity doesn't seem to be right answer yet,
at least for LCIS. It seems to us that one can apply the combinatorial bound of Duraj on the number
of significant pairs, and combine it with our approach, to achieve an even better complexity.
However, as this doesn't seem to result in a clean bound of (say) $\Oh(n^{2}/\log n)$ yet, we
leave determining the exact complexity for future work.

\bibliography{biblio}

\end{document}